\documentclass[10pt]{article}
\usepackage{amsmath}
\usepackage{amssymb}
\usepackage{amsthm}
\usepackage{enumerate}
\usepackage{epsf}
\usepackage{psfrag}
\usepackage{mathrsfs}
\usepackage{hyperref}
\usepackage[dvips]{graphicx}
\DeclareGraphicsExtensions{.eps,.art,.ART,.ps}
\newcommand{\bbR}{\mathbb{R}}      
\newcommand{\bbZ}{\mathbb{Z}}      

\newcommand{\grad}{\operatorname{grad}}

\newcommand{\dive}{\operatorname{div}}

\newcommand{\cotan}{\operatorname{cotan}}
\newtheorem{Thm}{Theorem}[section]
\newtheorem{Prop}[Thm]{Proposition}

\newtheorem{Def}[Thm]{Definition}

\newtheorem{Remark}[Thm]{Remark}

\begin{document}
\title{Elastic shocks in relativistic rigid rods and balls}
\author{Jo\~{a}o L.~Costa$^{1,2}$ and Jos\'e Nat\'ario$^{2}$\\ \\
{\small $^1$ ISCTE - Instituto Universit\'{a}rio de Lisboa, Lisboa, Portugal}\\
{\small $^2$ CAMGSD, Departamento de Matem\'{a}tica, Instituto Superior T\'{e}cnico,}\\
{\small Universidade de Lisboa, Portugal}
}
\date{}
\maketitle
\begin{abstract}
We study the free boundary problem for the``hard phase" material introduced by Christodoulou in \cite{Christodoulou95}, both for rods in $(1+1)$-dimensional Minkowski spacetime and for spherically symmetric balls in $(3+1)$-di\-men\-sio\-nal Minkowski spacetime. Unlike Christodoulou, we do not consider a ``soft phase", and so we regard this material as an elastic medium, capable of both compression and stretching. We prove that shocks must be null hypersurfaces, and derive the conditions to be satisfied at a free boundary. We solve the equations of motion of the rods explicitly, and we prove existence of solutions to the equations of motion of the spherically symmetric balls for an arbitrarily long (but finite) time, given initial conditions sufficiently close to those for the relaxed ball at rest. In both cases we find that the solutions contain shocks if and only if the pressure or its time derivative do not vanish at the free boundary initially. These shocks interact with the free boundary, causing it to lose regularity.
\end{abstract}
\tableofcontents
%
%
%
%
\section{Introduction}\label{section0}
In \cite{Christodoulou95}, Christodoulou introduced a two-phase fluid model aiming to describe stellar collapse with the possible formation of a neutron star. This model comprised a dust ``soft phase", representing the uncompressed stellar material, and a fluid ``hard phase", representing stellar material compressed to nuclear densities. The transition between the two phases happened at a certain threshold density, and it was in fact the main focus of research within this model (see \cite{Christodoulou96, Christodoulou96b, CL16}; see also \cite{FS19} for results concerning equilibrium states).

In this work, we study the free boundary problem for the ``hard phase" material surrounded by vacuum, both in $(1+1)$-dimensional and in $(3+1)$-dimensional Minkowski spacetime under spherical symmetry. Unlike Christodoulou, we do not consider a ``soft phase", and so we regard this material as an elastic medium, capable of both compression and stretching (that is, we allow the pressure to become negative). To make this distinction clear, we will refer to our model as a rigid elastic fluid. This problem is arguably the simplest relativistic model of an extended body, and may also help model neutron star crusts, which are expected to display elastic behavior (see for instance \cite{CH08}; several relativistic elastic models for compact objects have been considered in the literature, both as to the existence of equilibrium states \cite{Park00, KS04, FK07, ABS08, ABS09, BCV10, AC14} and as to their dynamical evolution \cite{Magli97, Magli98, BW07, AOS16, BM17}). A drawback of discarding the ``soft phase" is that we do have to allow for the possible occurrence of material singularities, where the stretching becomes infinite and the elastic medium tears itself apart. As we shall see, this cannot happen in the case of rods in $(1+1)$-dimensional Minkowski spacetime, but it may happen for spherically symmetric balls in $(3+1)$-dimensional Minkowski spacetime (see Appendix~\ref{appendixA}).

We now describe in detail the results in this paper. In Section~\ref{section1} we review the model for an irrotational rigid elastic fluid introduced as the ``hard phase" material in \cite{Christodoulou95}. In Section~\ref{section2} we show that the Rankine-Hugoniot conditions for the conservation of particles, energy and momentum across a hypersurface require that shocks, defined as hypersurfaces where the material's density, pressure and velocity are discontinuous (see Remark~\ref{Remark_on_shocks}), be null hypersurfaces. We also derive the constraints imposed by the Rankine-Hugoniot conditions at a free boundary. In Section~\ref{section4} we explicitly determine the motion of a rigid elastic rod, modeled as a segment of rigid elastic fluid moving in $(1+1)$-dimensional Minkowski spacetime. In Section~\ref{section5} we study the motion of a rigid elastic ball, modeled as a ball of rigid elastic fluid moving in $(3+1)$-dimensional Minkowski spacetime, under the assumption of spherical symmetry. We prove the existence of solutions to the equations of motion for an arbitrarily long (but finite) time, given initial conditions sufficiently close to those for the relaxed ball at rest. In both cases we find that the solutions contain shocks if and only if the pressure or its time derivative do not vanish initially at the free boundary. 
%
%
\section{Irrotational rigid elastic fluid model}\label{section1}
In this section we review the model for an irrotational rigid elastic fluid introduced as the ``hard phase" in \cite{Christodoulou95}. We start by considering the energy-momentum tensor of a massless scalar field $\phi:M \to \bbR$ on a given time-oriented Lorentzian manifold $(M,g)$,
\begin{equation}
T_{\mu\nu} = \partial_\mu \phi \, \partial_\nu \phi - \frac12 (\partial_\alpha \phi \, \partial^\alpha \phi) \, g_{\mu\nu}.
\end{equation}
If $\grad\phi$ is timelike and future-pointing then there exists a unit timelike future-pointing vector field $U$ and a positive function $n$ such that
\begin{equation}
\partial_\mu \phi = n U_\mu,
\end{equation}
and so
\begin{equation}
T_{\mu\nu} = n^2 U_\mu U_\nu + \frac12 n^2 g_{\mu\nu}.
\end{equation}
This is the energy-momentum tensor of a perfect fluid with $4$-velocity $U$,
\begin{equation}
T_{\mu\nu} = (\rho + p) U_\mu U_\nu + p g_{\mu\nu},
\end{equation}
whose density $\rho$ and pressure $p$ satisfy
\begin{equation}
\rho = p = \frac{n^2}2.
\end{equation}
The equation $p = \rho$ is the equation of state for a {\em stiff fluid}: the speed of sound (see \cite{Christodoulou07}) is equal to the speed of light,
\begin{equation}
\frac{dp}{d\rho}=1,
\end{equation}
but the pressure is always positive, that is, the medium is always compressed and never relaxed. To allow for the possibility of zero or even negative pressure, we consider the modified energy-momentum tensor
\begin{align}
T_{\mu\nu} = \partial_\mu \phi \, \partial_\nu \phi - (\frac12 \partial_\alpha \phi \, \partial^\alpha \phi) \, g_{\mu\nu} - \frac12 g_{\mu\nu} = n^2 U_\mu U_\nu + \frac12 (n^2-1) g_{\mu\nu}.
\end{align}
We can regard it as the energy-momentum tensor for a perfect fluid with density $\rho$ and pressure $p$ satisfying
\begin{equation}
\begin{cases}
\rho = \frac12 (n^2 + 1) \\
p =  \frac12 (n^2 - 1)
\end{cases},
\end{equation}
that is, with equation of state
\begin{equation}
p = \rho  - 1.
\end{equation}
This is the ``hard phase" material introduced by Christodoulou in \cite{Christodoulou95}: the speed of sound is still equal to the speed of light, but the material can either be compressed ($n>1$), relaxed ($n=1$) or streched ($n<1$).\footnote{Note that in~\cite{Christodoulou95} this material was not allowed to become stretched, but instead transitioned to the ``soft phase" when $n=1$. When the material is allowed to stretch, the relaxed state $n=1$ is indeed a minimum of the energy per particle $\frac{\rho}{n} = \frac12 \left(n+\frac1n\right)$.} The rest density when relaxed is $\rho=1$, which we can always assume by choice of units. We shall call this material a {\em rigid elastic fluid}: fluid because its internal energy only depends on its number density $n$ (see Appendix~\ref{appendixB}), elastic because it allows for stretching, and rigid because the speed of sound equals the speed of light (following the tradition in \cite{HM52, McCrea52, Bento85, N14}).

The equation of motion for the scalar field is still the massless wave equation,
\begin{equation} \label{energymomentumconservation}
\nabla^\mu T_{\mu\nu} = 0 \Leftrightarrow \Box \phi = 0.
\end{equation}

If we define the vector field
\begin{equation}
I_{\mu} = n U_\mu = \partial_\mu \phi,
\end{equation}
then the equation of motion for the scalar field implies
\begin{equation} \label{numberconservation}
\nabla^\mu I_{\mu} = \Box \phi = 0,
\end{equation}
that is, this vector is a conserved current. In fact, it can be interpreted as the {\em number density vector}, associated to the conserved number of particles in the rigid elastic fluid. Therefore, if $\grad \phi$ ever becomes null (which may perfectly well happen for the wave equation), that is, if ever $n \to 0$, this should be interpreted as a {\em material singularity}, corresponding to an infinite streching (see Appendix~\ref{appendixA}). Note that the equation of state $p=\rho-1$ implies the dominant energy condition if and only if
\begin{equation}
\rho + p \geq 0 \Leftrightarrow p \geq - \frac12,
\end{equation}
which is precisely the condition $n^2 \geq 0$.

Finally, note that the Frobenius theorem states that locally the condition
\begin{equation}
U = \frac1n d\phi
\end{equation}
for some functions $n, \phi$ is equivalent to
\begin{equation}
U \wedge dU = 0.
\end{equation}
In other words, taking $U$ to be proportional to $\grad \phi$ is equivalent to taking $U$ to be irrotational (that is, to have zero vorticity). This happens automatically in $(1+1)$ dimensions, and also in spherical symmetry: for instance, in the later case we have
\begin{equation}
U = U_0(t,r) dt + U_r(t,r) dr.
\end{equation}

We summarize the results in this section as follows:

\begin{Prop}
Any solution $\phi:M\to \bbR$ of the wave equation on a time-oriented Lorentzian manifold $(M,g)$ can be thought of as representing an irrotational rigid elastic fluid in the region where its gradient is timelike and future-pointing. The fluid's number density vector is precisely $\grad \phi$, and its $4$-velocity $U$, density $\rho$ and pressure $p$ are given (in appropriate units) by
\begin{align}
& U = \frac1{n} \grad \phi; \\
& \rho = \frac12 (n^2 + 1); \\
& p =  \frac12 (n^2 - 1),
\end{align}
where $n^2 = -\langle \grad \phi, \grad \phi \rangle$.
\end{Prop}
%
%
\section{Rankine-Hugoniot conditions}\label{section2}
In this section we show that the Rankine-Hugoniot conditions for the conservation of particles, energy and momentum across a shock (defined as a hypersurface where the material's density, pressure and velocity are discontinuous, see Remark~\ref{Remark_on_shocks}) of an irrotational rigid elastic fluid require the shock to be a null hypersurface, across which the scalar field only has to be continuous. Moreover, we derive the constraints imposed by the Rankine-Hugoniot conditions at a free boundary.

Let $\Sigma$ be a two-sided $C^2$ hypersurface, and suppose that $\phi$ is a solution of the wave equation on either side of $\Sigma$. Assume that $\phi$ is continuous across $\Sigma$ while being $C^2$ on either side, up to and including $\Sigma$ (meaning that the restriction of $\phi$ to each side of $\Sigma$ can be extended to a $C^2$ function defined on an open neighborhood of $\Sigma$), i.e.~that $\phi$ is continuous and sectionally $C^2$. The Rankine-Hugoniot conditions for the rigid elastic fluid represented by $\phi$ require the contractions of both the number density vector $I_\mu$ and the energy-momentum tensor $T_{\mu\nu}$ with the normal to $\Sigma$ to be continuous (including the points where $\Sigma$ is null).\footnote{Recall that these conditions are obtained by imposing the integral form of the conservation equations \eqref{energymomentumconservation} and \eqref{numberconservation} on ``thin'' open neighborhoods of points in $\Sigma$.} 

Consider first the case when $\Sigma$ is spacelike, and let $\{E_0,\ldots,E_n\}$ be a local orthonormal frame such that $E_0$ is the unit (timelike) future-pointing normal to $\Sigma$. We have
\begin{equation}
\langle I,E_0\rangle = E_0 \cdot \phi,
\end{equation}
and so the normal derivative $E_0 \cdot \phi$ must be continuous for the first Rankine-Hugoniot condition to hold. This is enough to ensure that the remaining conditions hold, since
\begin{equation}
T(E_0,E_0) = \frac12 (E_0 \cdot \phi)^2 + \frac12 (E_1 \cdot \phi)^2 + \ldots + \frac12 (E_n \cdot \phi)^2 + \frac12
\end{equation}
and
\begin{equation}
T(E_0,E_i) = (E_0 \cdot \phi)(E_i \cdot \phi)
\end{equation}
(our hypotheses imply that the tangential derivatives $E_i \cdot \phi$ are continuous). From the wave equation it is then clear that the second normal derivative of $\phi$ is also continuous across $\Sigma$, and so $\phi$ must in fact be $C^2$ across $\Sigma$, that is, there is no shock. Similar remarks apply to the case when $\Sigma$ is timelike.

If $\Sigma$ is null, however, things are different. Let $\{L,N,E_2,\ldots,E_n\}$ be a frame such that $L$ is a future-pointing null normal to $\Sigma$, $\{E_2, \ldots, E_n\}$ are orthonormal spacelike vector fields tangent to $\Sigma$ and $N$ is the future-pointing null vector field orthogonal to $\{E_2, \ldots, E_n\}$ and normalized by $\langle L, N\rangle = -1$. Since the null normal $L$ is also tangent to $\Sigma$, both
\begin{equation}
\langle I,L\rangle = L \cdot \phi
\end{equation}
and the components
\begin{equation}
T(L,L) = (L \cdot \phi)^2
\end{equation}
and
\begin{equation}
T(L,E_i) = (L \cdot \phi)(E_i \cdot \phi)
\end{equation}
are automatically continuous across $\Sigma$. The matrix of the metric in this frame is
\begin{equation}
(g_{\mu\nu}) =
\left(
\begin{matrix}
0 & -1 & 0 & \ldots & 0 \\
-1 & 0 & 0 & \ldots & 0 \\
0 & 0 & 1 & \ldots & 0 \\
\ldots & \ldots & \ldots & \ldots & \ldots \\
0 & 0 & 0 & \ldots & 1 \\
\end{matrix}
\right)
= (g_{\mu\nu})^{-1},
\end{equation}
and so
\begin{equation}
\langle d\phi, d\phi\rangle = -2(L \cdot \phi)(N \cdot \phi) + (E_2 \cdot \phi)^2 + \ldots + (E_n \cdot \phi)^2.
\end{equation}
Substituting in the expression of the energy-momentum tensor yields
\begin{equation}
T(L,N) = \frac12 (E_2 \cdot \phi)^2 + \ldots + \frac12 (E_n \cdot \phi)^2 + \frac12,
\end{equation}
which again is continuous across $\Sigma$. In other words, in the case when $\Sigma$ is null the Rankine-Hugoniot conditions are satisfied whenever $\phi$ is continuous across $\Sigma$. The derivatives transverse to $\Sigma$ do not have to be continuous, and $\phi$ is a weak solution of the wave equation, by a standard argument that we now recall: if $\varphi$ is a test function with sufficiently small support $K$ intersecting $\Sigma$ such that $K \setminus \Sigma$ has two connected components $K^+$ and $K^-$ then
\begin{align}
\int_{K^\pm} \phi \, \Box \varphi & = \int_{K^\pm} \dive(\phi \grad \varphi) - \int_{K^\pm} \left\langle \grad \phi, \grad \varphi \right\rangle \nonumber \\
& = \int_{\partial K^\pm} \phi (\pm L \cdot \varphi) - \int_{K^\pm} \dive(\varphi \grad \phi) + \int_{K^\pm} \varphi \, \Box \phi \\
& = \int_{\partial K^\pm} \phi (\pm L \cdot \varphi) - \int_{\partial K^\pm} \varphi (\pm L \cdot \phi), \nonumber
\end{align}
and so
\begin{align}
\int_M \phi \, \Box \varphi & = \int_{K^+} \phi \, \Box \varphi + \int_{K^-} \phi \, \Box \varphi \nonumber \\
& = \int_{\Sigma \cap K} [\phi] (L \cdot \varphi) - \int_{\Sigma \cap K} \varphi [L \cdot \phi] = 0,
\end{align}
where $[\cdot]$ denotes the jump across $\Sigma$.

We have then proven the following result:

\begin{Prop}
Let $\Sigma$ be a two-sided $C^2$ hypersurface, and suppose that $\phi$ is a solution of the wave equation on either side of $\Sigma$. Assume that $\phi$ is continuous across $\Sigma$ while being $C^2$ on either side (up to and including $\Sigma$). If the Rankine-Hugoniot conditions for the irrotational rigid elastic fluid represented by $\phi$ hold then either (i) $\Sigma$ is null, with $\phi$ a weak solution having discontinuous transverse derivatives, or (ii) $\phi$ is $C^2$ across $\Sigma$, in which case $\Sigma$ is not a shock.
\end{Prop}

Another important result related to the Rankine-Hugoniot conditions is the following:

\begin{Prop} \label{RHfree}
Let $\Sigma$ be a $C^1$ {\em free boundary} of the irrotational rigid elastic fluid, that is, a timelike $C^1$ hypersurface separating the fluid (represented by a solution $\phi$ of the wave equation of class $C^2$ up to and including $\Sigma$) from the vacuum. If the Rankine-Hugoniot conditions hold then (i) $\grad \phi$ is tangent to $\Sigma$ and (ii) $\langle \grad \phi, \grad \phi \rangle = -1 \Leftrightarrow p=0$ on $\Sigma$.
\end{Prop}

\begin{proof}
The Rankine-Hugoniot conditions require, in this case, that the contractions of both the number density vector $I_\mu$ and the energy-momentum tensor $T_{\mu\nu}$ of the fluid with the spacelike normal to $\Sigma$ vanish, since both the number density vector and the energy-momentum tensor are zero on the exterior. This is easily seen to be equivalent to $\grad \phi$ being tangent to $\Sigma$ plus the vanishing of the pressure $p = \frac12 \left(- \langle \grad \phi, \grad \phi \rangle - 1 \right)$ on $\Sigma$.
\end{proof}

In view of the results above, we finish this section with the following definition:

\begin{Def} \label{motion}
Let $(M,g)$ be a time-oriented Lorentzian manifold. An {\em irrotational motion} of a rigid elastic fluid is a pair $(\Omega,\phi)$, where $\Omega \subset M$ is a region with a sectionally $C^1$ timelike boundary and $\phi: \overline{\Omega} \to \bbR$ is a continuous, sectionally $C^2$ weak solution of the wave equation such that $\grad \phi$ is timelike, tangent to $\partial\Omega$ and satisfying $\langle \grad \phi, \grad \phi \rangle = -1$ on the $C^1$ components of $\partial\Omega$. We also require that the number of shocks intersecting any compact subregion of $\Omega$ is finite, where by {\em shocks} we mean null hypersurfaces $\Sigma$ on which $\phi$ is continuous but not $C^2$ (although it is $C^2$ on each side of $\Sigma$, up to and including $\Sigma$).
\end{Def}

\begin{Remark} \label{Remark_on_shocks}
Recall that a {\em shock} is usually defined as a hypersurface of discontinuity across which there is a nonzero flux of particles, by opposition to a {\em contact discontinuity} \cite{CF48}; when the discontinuity hypersurfaces are null, as above, this is always the case. Note, however, that the shocks studied in this paper are necessarily reversible, because the wave equation can be solved backwards. In fact, it can be shown that there is no jump in entropy across these null shocks (see \cite{Christodoulou95}).
\end{Remark}

\begin{Remark}
As mentioned above, we will only consider situations where the motions are automatically irrotational. Consequently, we will speak simply of ``motions".
\end{Remark}

%
%
%
\section{Rigid elastic rod}\label{section4}
In this section we will study the motion of a rigid elastic rod, modeled as a segment of rigid elastic fluid moving in $(1+1)$-dimensional Minkowski spacetime. Note that in $(1+1)$ dimensions all elastic energy-momentum tensors correspond to perfect fluids (see for instance \cite{N14}), and so this is the unique possible choice for a rigid elastic rod. Examples of motions of these elastic rods were already given in \cite{HM52, McCrea52, BF03, N14}.

Define $T=-\phi$, so that $T$ is a time coordinate. Let $X$ be a coordinate along a line of constant $\phi$, and extend it to all such lines by the flow of $\grad \phi$. Note that $X$ is a comoving coordinate, that is, the lines of constant $X$ are the worldlines of the fluid's particles. Since
\begin{equation}
\left\langle dT, dT \right\rangle = \left\langle d\phi, d\phi \right\rangle = -n^2,
\end{equation}
the Minkowski metric is written in these coordinates as
\begin{equation}
ds^2 = \frac1{n^2} \left( - dT^2 + A^2 dX^2 \right)
\end{equation}
for some positive function $A$. We have
\begin{equation}
\Box T = 0 \Leftrightarrow \partial_\mu \left(\sqrt{-g} \, \partial^\mu T \right) = 0 \Leftrightarrow \frac{\partial A}{\partial T} = 0,
\end{equation}
that is, $A = A(X)$. We can therefore rescale the coordinate $X$ to make $A = 1$ and write the metric in the form
\begin{equation} \label{metric}
ds^2 = \frac1{n^2} \left( - dT^2 + dX^2 \right).
\end{equation}
Since this is the Minkowski metric, the curvature must be zero. If we set
\begin{equation}
n = e^\psi,
\end{equation}
the condition for zero curvature is
\begin{equation}
R_{0101} = 0 \Leftrightarrow - \frac{\partial^2 \psi}{\partial T^2} + \frac{\partial^2 \psi}{\partial X^2} = 0,
\end{equation}
that is, $\psi$ satisfies the one-dimensional wave equation, whence
\begin{equation}
\psi(T,X) = f(T-X) + g(T+X).
\end{equation}
Note that the Minkowski metric \eqref{metric} can then be written as
\begin{equation}
ds^2 = e^{-2\psi} \left( - dT^2 + dX^2 \right) = - e^{-2f(U)-2g(V)} dU dV,
\end{equation}
where
\begin{equation} \label{coord1}
\begin{cases}
U = T - X \\
V = T + X
\end{cases}
\Leftrightarrow
\begin{cases}
T = \frac12(U+V) \\
X = \frac12(V-U)
\end{cases},
\end{equation}
and so the coordinate transformation from the comoving coordinates $(T,X)$ to the usual Lorentz inertial coordinates $(t,x)$ can be obtained by defining
\begin{equation} \label{coord2}
\begin{cases}
u = \int e^{-2f(U)} dU \\
v = \int e^{-2g(V)} dV
\end{cases}
\end{equation}
and setting
\begin{equation} \label{coord3}
\begin{cases}
t = \frac12(u+v) \\
x = \frac12(v-u)
\end{cases}.
\end{equation}
If there are no shocks then $T = -\phi$ is a function of class $C^2$, and so the coordinate transformation \eqref{coord1}-\eqref{coord3} is also $C^2$. As a consequence, the Minkowski metric \eqref{metric} is of class $C^1$ when written in the comoving coordinates $(T,X)$, which agrees with $n^2$ being a $C^1$ function of the inertial coordinates $(t,x)$. To accomodate shocks, however, we must allow for the possibility that $n^2$ may be discontinuous across null lines, which is the same as allowing $f$ and $g$ to have discontinuities. In this case, the coordinate transformation \eqref{coord1}-\eqref{coord3} is merely continuous across the shocks, and the metric \eqref{metric} (or equivalently $n^2$) is even discontinuous. Following Definition~\ref{motion}, we will nevertheless require $f$ and $g$ to be sectionally $C^1$,\footnote{By this we mean that $f$ and $g$ are functions of class $C^1$ except for a set isolated discontinuities, where their restrictions to sufficiently small left or right open neighborhoods can be extended to $C^1$ functions defined on full open neighborhoods.} so that the coordinate transformation \eqref{coord1}-\eqref{coord3} is sectionally $C^2$, and $n^2$ is sectionally $C^1$ (as a function of both sets of coordinates).

\begin{Remark}
It is interesting to note that conformal coordinates in the Euclidean plane must be $C^2$ (in fact analytic), since weak solutions of the Laplace equation are automatically analytic (see for instance \cite{DK81}). In other words, if one takes local coordinates in the Euclidean plane which are continuous, sectionally $C^2$ functions of the usual Cartesian coordinates with discontinuous derivatives, then the Euclidean metric written in the new coordinates is never conformal to the Euclidean metric.
\end{Remark}

If we assume that the rod's free endpoints correspond to $X=0$ and $X=L$, the zero pressure condition (which applies to both endpoints) implies
\begin{equation}
n(T,0) = 1 \Leftrightarrow f(T) + g(T) = 0,
\end{equation}
whence
\begin{equation}
\psi(T,X) = f(T-X) - f(T+X) = f(U) - f(V),
\end{equation}
and
\begin{equation}
n(T,L) = 1 \Leftrightarrow f(T-L) - f(T+L) = 0,
\end{equation}
that is, $f$ is periodic with period $2L$. In particular, $f$ is bounded, and the kind of material singularities corresponding to infinite stretching do not occur in this one-dimensional model.

Let us assume that
\begin{equation}
n(0,X) \in C^1\left(\left]0,L\right[\right),
\end{equation}
or, equivalently,
\begin{equation}
\psi(0,X) = f(-X) - f(X) \in C^1\left(\left]0,L\right[\right).
\end{equation}
Then the odd part of $f$ in the interval $[-L,L]$ is in $C^1\left(\left]-L,0\right[\cup\left]0,L\right[\right)$. If in addition we assume that
\begin{equation}
\frac{\partial n}{\partial T}(0,X)  \in C^0\left(\left]0,L\right[\right),
\end{equation}
or, equivalently,
\begin{equation}
\frac{\partial\psi}{\partial T}(0,X) = f'(-X) - f'(X)  \in C^0\left(\left]0,L\right[\right),
\end{equation}
then the even part of $f$ in the interval $[-L,L]$ is in $C^1\left(\left]-L,0\right[\cup\left]0,L\right[\right)$. We conclude that in this case $\psi$ is of class $C^1$ except possibly on the set
\begin{equation} \label{shockwave}
SW = \{(T,X) \in \bbR \times \left] 0,L \right[ : U = nL \text{ or } V=nL \text{ for some } n \in \bbZ\}.
\end{equation}
At each of these values it may happen that $f$ is not continuous (and in that case $n(0,X)$ does not tend to one at either $X=0$ or $X=L$) or $f'$ is not continuous (and in that case $\frac{\partial n}{\partial T}(0,X)$ does not tend to zero at either $X=0$ or $X=L$). So one obtains shocks if and only if initially either the rod is stretched or compressed at the endpoints, or if its stretching at the endpoints is changing in time, that is, if and only if the initial data is in tension with the boundary conditions.

We may summarize the conclusions above as follows.

\begin{Thm}
Consider a rigid elastic rod, modeled as a segment of rigid elastic fluid moving in $(1+1)$-dimensional Minkowski spacetime (according to Definition~\ref{motion}). It is possible to define coordinates $(T,X)$ on the spacetime region corresponding to the rod such that the lines of constant $X$ are the worldlines of the rod's particles and the Minkowski metric reads
\begin{equation}
ds^2 = \frac1{n^2} \left( - dT^2 + dX^2 \right).
\end{equation}
If the rod's endpoints are given by $X=0$ and $X=L$ then each motion is determined by a $2L$-periodic, sectionally $C^1$ function $f:\bbR \to \bbR$ such that
\begin{equation}
n(T,X) = e^{f(T-X)-f(T+X)}.
\end{equation}
In particular, the motion is defined for all times, and the coordinates $(T,X)$ are continuous, sectionally $C^2$ functions of the inertial coordinates $(t,x)$, with discontinuous derivatives along the null lines (shocks) determined by the discontinuities of $f$, provided they exist. If $n(0,X)$ and $\frac{\partial n}{\partial T}(0,X)$ are of class $C^1$ and $C^0$ in the open interval $\left]0,L\right[$, respectively, then all shocks either emanate from the rod's endpoints at $T=0$ or correspond to reflections of these shocks at the free boundary. These shocks are present if and only if either $n(0,X)$ does not tend to $1$ or $\frac{\partial n}{\partial T}(0,X)$ does not tend to $0$ at the endpoints $x=0$ or $x=L$.
\end{Thm}

As an example, let us take
\begin{equation}
f(X)=
\begin{cases}
0 \qquad \text{ for } \quad -L < X < 0, \\
\varepsilon \qquad \text{ for } \quad 0 < X < L,
\end{cases}
\end{equation}
for $\varepsilon > 0$. Then $\psi$ is sectionally constant, assuming the values depicted in Figure~\ref{rod}. Note that
\begin{equation}
\begin{cases}
\psi = 0 \quad & \Rightarrow \quad p = 0; \\
\psi = \varepsilon & \Rightarrow \quad p = \frac12(e^{2\varepsilon} - 1) > 0; \\
\psi = -\varepsilon & \Rightarrow \quad p = \frac12(e^{-2\varepsilon} - 1) < 0,
\end{cases}
\end{equation}
and so this solution represents a rod that is initially stretched and is then released. Shock waves propagate at the speed of light from the free ends inwards, causing the rod to relax; as the shock waves cross, they compress the rod, and the process repeats periodically.

\begin{figure}[h!]
\begin{center}
\psfrag{T}{$T$}
\psfrag{X}{$X$}
\psfrag{0}{$0$}
\psfrag{L}{$L$}
\psfrag{2L}{$2L$}
\psfrag{e}{$\varepsilon$}
\psfrag{-e}{$-\varepsilon$}
\psfrag{t}{$t$}
\psfrag{x}{$x$}
\psfrag{eL}{$e^{\varepsilon}L$}
\psfrag{cL}{$\cosh({\varepsilon})L$}
\psfrag{2cL}{$2\cosh({\varepsilon})L$}
\epsfxsize=.4\textwidth
\leavevmode
\epsfbox{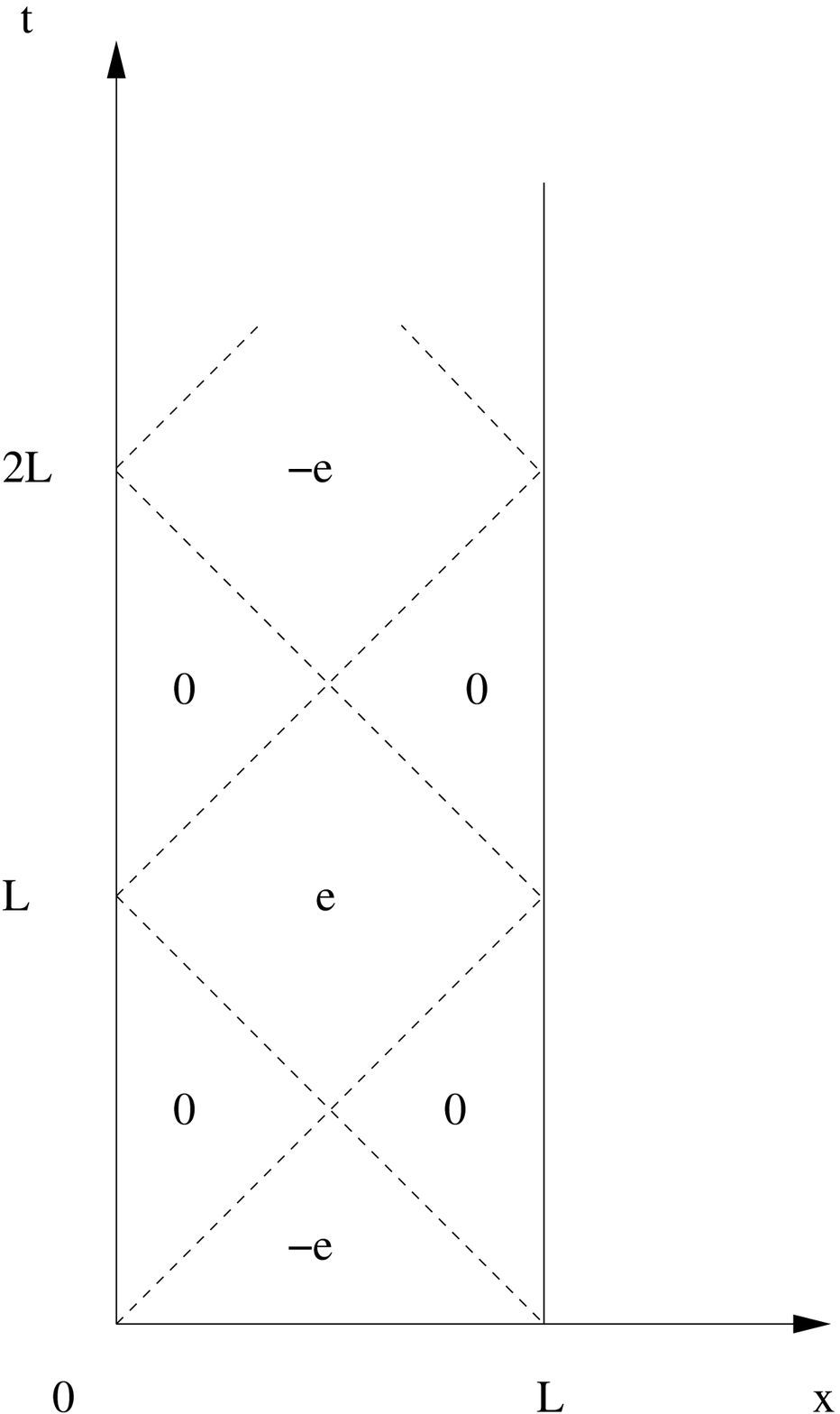}
\hspace{1cm}
\epsfxsize=.5\textwidth
\leavevmode
\epsfbox{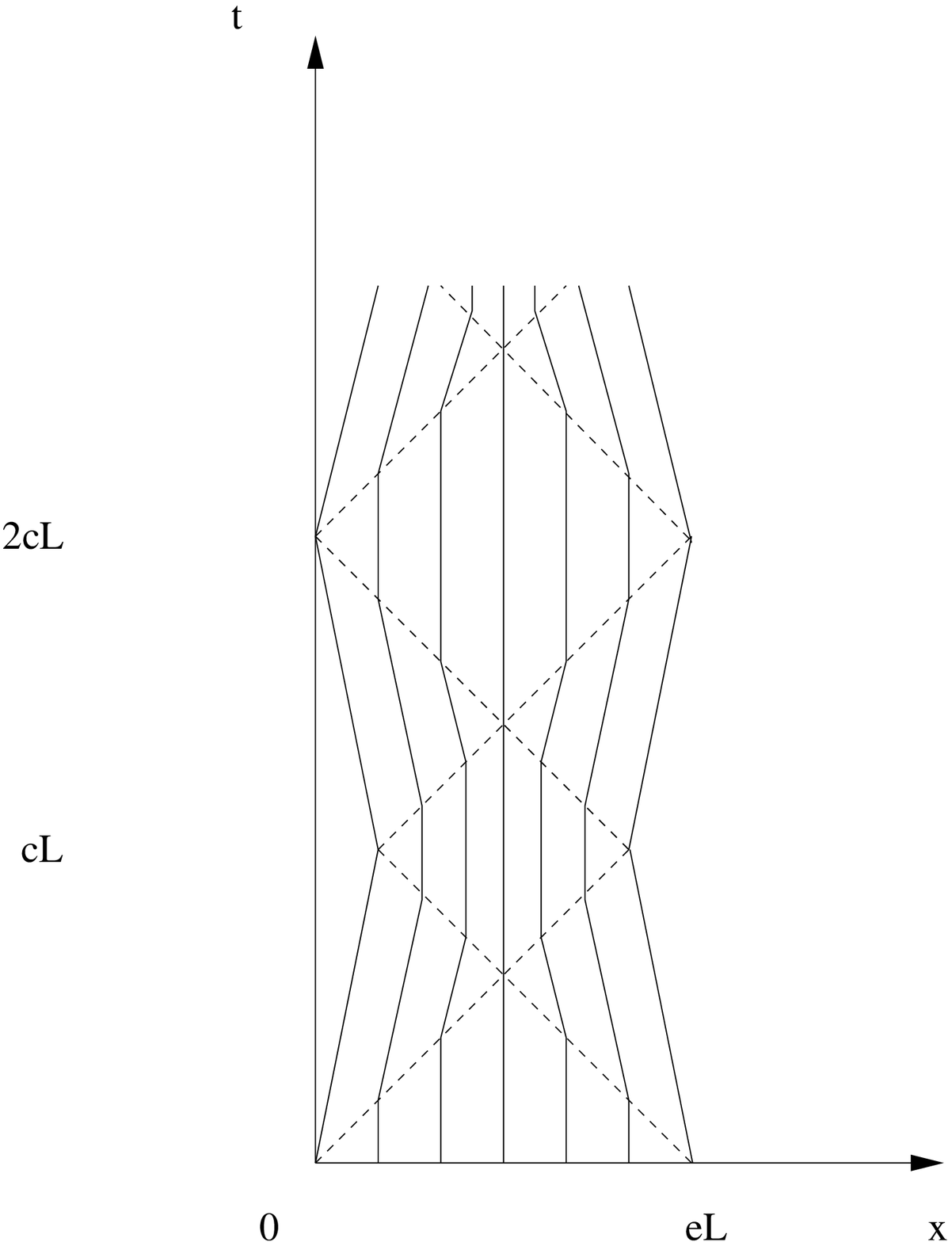}
\end{center}
\caption{Example of a rigid elastic rod motion. Left: in comoving coordinates, with the different values of $\psi$ displayed. Right: in inertial coordinates.}\label{rod}
\end{figure}

To get a picture of the rod in inertial coordinates $(t,x)$, we note that because $n^2$ is sectionally constant the worldlines of the rod's particles are sectionally geodesic. By symmetry, the central worldline (corresponding to $X=\frac{L}2$) is a full geodesic, and we choose the inertial coordinates $(t,x)$ such that this geodesic is a line of constant $x$.\footnote{Note that these coordinates are not the coordinates $(t,x)$ given by the coordinate transformation \eqref{coord1}-\eqref{coord3}, but are of course related to these by a Lorentz transformation.} From the picture in comoving coordinates $(T,X)$ it is clear that initially the rod is at rest with respect to this central line (the distance between worldlines remains initially constant), and also that the initial length of the rod is its proper length, which, in view of \eqref{metric}, is given by $e^{\varepsilon}L$. Moreover, the rod is again at rest with respect to the central line after a proper time interval (i.e.~an inertial time interval) $e^{\varepsilon} \frac{L}2 + e^{-\varepsilon} \frac{L}2 = \cosh(\varepsilon) L$, and its length is then $e^{-\varepsilon}L$. We conclude that the endpoints of the rod must have moved towards the center with speed
\begin{equation}
\frac{\frac12\left(e^{\varepsilon}L-e^{-\varepsilon}L\right)}{\cosh(\varepsilon) L} = \tanh(\varepsilon).
\end{equation}
By symmetry, this is also the speed at which the endpoints of the rod move away from the center afterwards, and therefore the rod's motion is as depicted in Figure~\ref{rod}.

%
%
\section{Rigid elastic ball}\label{section5}
In this section we will study the motion of a rigid elastic ball, modeled as a ball of rigid elastic fluid moving in $(3+1)$-dimensional Minkowski spacetime, under the assumption of spherical symmetry. These motions can also be obtained from other elastic models which are not fluids (see Appendix~\ref{appendixB}).

Under spherical symmetry, the solution of the wave equation in $(3+1)$-dimensional Minkowski spacetime is given by
\begin{equation}
\Box \phi=0 \Leftrightarrow \phi(t,r)=\frac{f(t-r)-g(t+r)}{r},
\end{equation}
where $r$ is the usual spherical coordinate. If $\phi$ is to be continuous at $r=0$ then we must have $f(t)-g(t)=0$, whence
\begin{equation}
\phi(t,r)=\frac{f(t-r)-f(t+r)}{r}.
\end{equation}
Note that the limit
\begin{equation}
\lim_{r\to 0}\phi(t,r) = -2f'(t)
\end{equation}
exists and is continuous whenever $f$ is $C^1$. In this case, we have
\begin{equation}
\partial_r\phi(t,r)=\frac{-rf'(t-r)-rf'(t+r)-f(t-r)+f(t+r)}{r^2},
\end{equation}
and the limit
\begin{equation}
\lim_{r\to 0}\partial_r\phi(t,r) = \lim_{r\to 0} \frac{rf''(t-r)-rf''(t+r)}{2r} = 0
\end{equation}
exists whenever $f$ is $C^2$. More generally, one can prove that if $f$ is $C^n$ then
\begin{equation}
\lim_{r\to 0}\partial_r^{n-1}\phi(t,r)=
\begin{cases}
0 \qquad \qquad \qquad \text{ if } n \text{ is even;} \\
-\frac{2}{n} f^{(n)}(t) \qquad \text{ if } n \text{ is odd.}
\end{cases}
\end{equation} The fact that $f$ has to be $C^n$ for $\phi$ to be $C^{n-1}$ is a consequence of the well known phenomenon of focusing and loss of derivatives in the wave equation.

In what follows we will take the function $f$ to be of class $C^2$. This means that $\grad\phi$ will be of class $C^1$ except possibly at the origin,\footnote{In view of this, the solutions described in this section only satisfy Definition~\ref{motion} away from the center of symmetry, i.e., to obtain a {\em motion of a rigid elastic fluid} we must exclude the worldline $\{r=0\}$ from $\Omega$. Alternatively, we could also extend Definition~\ref{motion} in order to accommodate this loss of regularity, but that would lead to technical complications that we prefer to avoid.} implying uniqueness of the fluid's worldlines with the possible exception of those reaching $r=0$. Now $r=0$ is clearly a worldline, and the fact that $\grad\phi$ is divergenceless guarantees that no other worldline can cross it; therefore we have uniqueness for all worldlines. In particular, the worldline representing the free boundary cannot reach $r=0$, and it is therefore a $C^2$ curve.

Solving the wave equation subject to the free boundary conditions is then equivalent to finding $C^2$ functions $f:\bbR \to \bbR$ and $r_b:\bbR \to \bbR^+$ such that
\begin{equation}
\begin{cases}
\displaystyle
\phi(t,r) = \frac{f(t-r)-f(t+r)}{r} \\ \\
\left((\partial_t\phi)^2-(\partial_r\phi)^2\right)(t,r_b(t)) = 1 \\ \\
\displaystyle
\frac{dr_b}{dt}(t)=-\frac{\partial_r\phi}{\partial_t\phi}(t,r_b(t))
\end{cases},
\end{equation}
where the last equation is the condition that $\grad \phi$ is tangent to the free boundary (see Proposition~\ref{RHfree}). Note that adding a constant to $f$ leaves $\phi$ unchanged, and adding a linear function to $f$ changes $\phi$ by an additive constant, with no impact on the physical solution. As an example, the solution corresponding to a relaxed ball of radius $R>0$,
\begin{equation}
\phi_{\rm rel}(t,r) = -t,
\end{equation}
say, can be obtained by taking
\begin{equation}
f(t)=\frac14 t^2, \qquad \qquad r_b(t)=R.
\end{equation}

In what follows it will be convenient to use the double null coordinates
\begin{equation}
\begin{cases}
u = t-r \\
v = t+r
\end{cases}
\Leftrightarrow
\begin{cases}
t = \frac12(u+v) \\
r = \frac12(v-u)
\end{cases},
\end{equation}
in terms of which the Minkowski metric is written
\begin{equation}
ds^2 = - du \, dv + r^2(u,v) \left( d\theta^2 + \sin^2\theta d\varphi^2 \right).
\end{equation}
The first $2 \times 2$ block of the inverse metric is then
\begin{equation}
\left(
\begin{matrix}
0 & -\frac12 \\
-\frac12 & 0
\end{matrix}
\right)^{-1}
=
\left(
\begin{matrix}
0 & -2 \\
-2 & 0
\end{matrix}
\right) ,
\end{equation}
yielding
\begin{equation}
\grad \phi = - 2 \, \partial_v \phi \, \frac{\partial}{\partial u} - 2 \, \partial_u \phi \, \frac{\partial}{\partial v}.
\end{equation}
In these coordinates, the problem becomes finding $C^2$ functions $f:\bbR \to \bbR$ and $u_b:\bbR \to \bbR$ such that
\begin{equation}
\begin{cases} \label{problemuv}
\displaystyle
\phi(u,v) = \frac{f(u)-f(v)}{r(u,v)} \\ \\
(\partial_u\phi \, \partial_v\phi)(u_b(v),v) = \frac14 \\ \\
\displaystyle
\frac{du_b}{dv}(v)=\frac{\partial_v\phi}{\partial_u\phi}(u_b(v),v)
\end{cases}.
\end{equation}
The solution corresponding to a relaxed ball of radius $R>0$,
\begin{equation}
\phi_{\rm rel}(u,v)=-\frac12 (u+v),
\end{equation}
is obtained by taking
\begin{equation}
f(t)=\frac14 t^2, \qquad \qquad u_b(v) = v - 2R.
\end{equation}

We will show that there exist motions near this relaxed state by applying the Implicit Function Theorem. In order to do that, let us define the function spaces
\begin{align}
& X=\{ \varphi \in C^2([-R,R]) : \varphi(R)=0 \} \, ; \nonumber \\
& Y=\{ \psi \in C^2([R,R+\delta]) : \psi(R)=0 \} \, ; \nonumber \\
& Z= C^1([R,R+\delta]) \, ,
\end{align}
equiped with the usual $C^2$ and $C^1$ norms, where $\delta > 0$ will be chosen later. The idea here is that $X$ is the space of possible deviations from the relaxed ball solution in the interval $[-R,R]$, that is, $f(t)=t^2/4 + \varphi(t)$ for $\varphi \in X$  and $t \in [-R,R]$, and $Y$ is the space of possible deviations in the interval $[R,R+\delta]$; we have fixed $f(R)=R^2/4$ to remove the ambiguity in the choice of $f$. The space $Z$ consists of the corresponding values of $n^2/4$ along the free boundary in the interval $[R,R+\delta]$ (see Figure~\ref{sphere}). 

Let us then consider the map $F:X \times Y \to Z$ defined as follows: given $(\varphi,\psi)\in X \times Y$, set
\begin{equation} \label{phimap}
\phi(u,v) = \phi_{\rm rel}(u,v) + \frac{\varphi(u)-\psi(v)}{r(u,v)}
\end{equation}
(defined on $[-R,R] \times [R,R+\delta]$), and let $u_b:[R,R+\delta] \to \bbR$ be the solution of the Cauchy problem
\begin{equation}
\begin{cases}
\displaystyle
\frac{du_b}{dv}(v)=\frac{\partial_v\phi}{\partial_u\phi}(u_b(v),v) \\
u_b(R)=-R
\end{cases}.
\end{equation}
If $(\varphi,\psi)$ is in a sufficiently small neighborhood $U \times V \subset X \times Y$ of the origin then $\grad \phi$ is close to $\frac{\partial}{\partial t}$, and so it is possible to choose $\delta > 0$ such that $u_b(v) \in [-R,R]$ for $v \in [R,R+\delta]$ (see Figure~\ref{sphere}). We then define
\begin{equation}
F(\varphi,\psi)(v) = (\partial_u\phi \, \partial_v\phi)(u_b(v),v).
\end{equation}
Note that $F$ is well defined on $U \times V$ if $\delta$ is sufficiently small, and satisfies
\begin{equation}
F(0,0)(v)=\frac14.
\end{equation}

\begin{figure}[h!]
\begin{center}
\psfrag{t}{$t$}
\psfrag{r}{$r$}
\psfrag{R}{$R$}
\psfrag{u=R}{\!\!$u=R$}
\psfrag{u=-R}{$u=-R$}
\psfrag{v=R}{\!\!\!\!$v=R$}
\psfrag{v=R+e}{$v=R+\delta$}
\psfrag{I}{$I$}
\psfrag{II}{$II$}
\psfrag{III}{$III$}
\psfrag{IV}{\!\!\!\!$IV$}
\epsfxsize=.4\textwidth
\leavevmode
\epsfbox{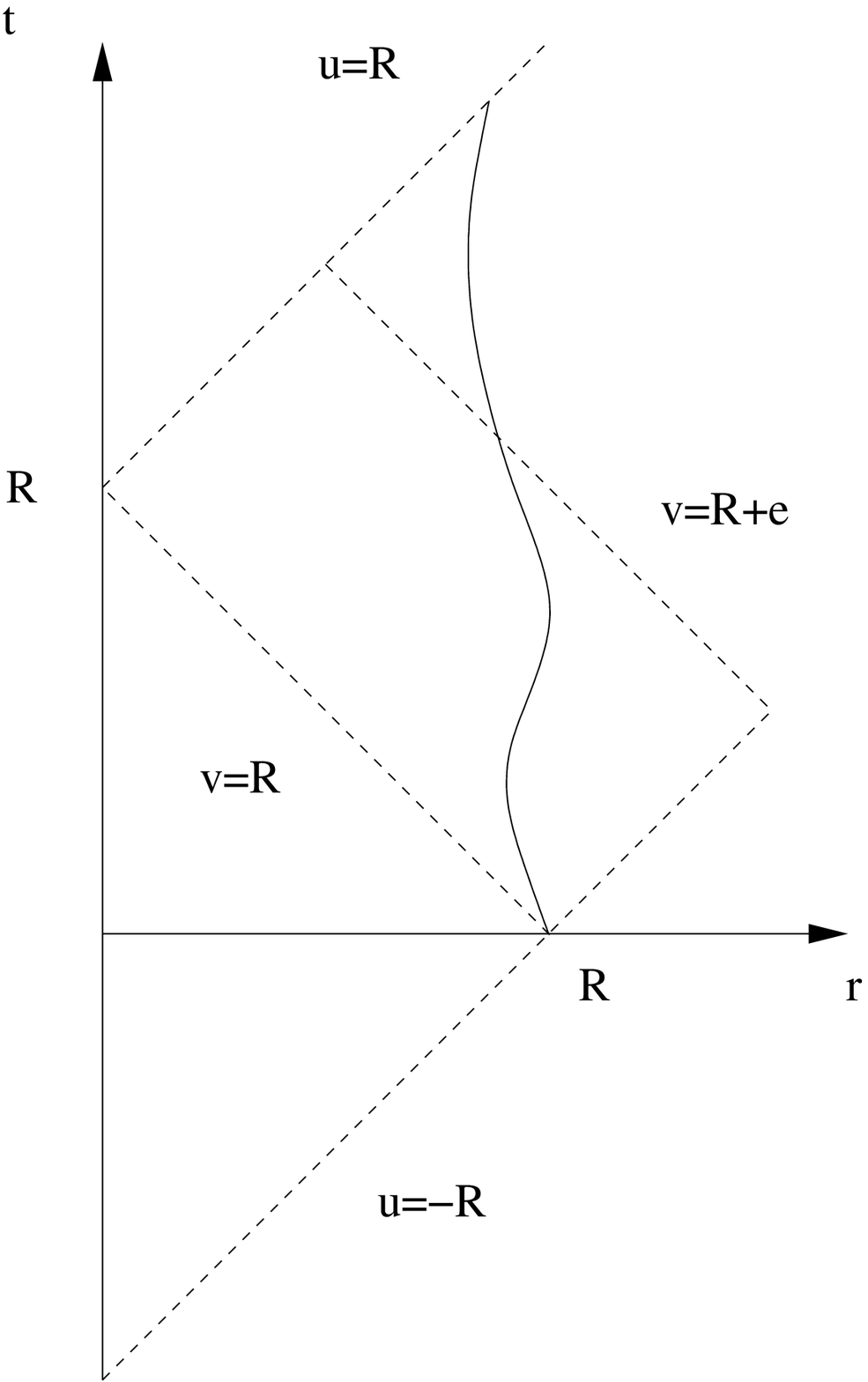}
\hspace{2cm}
\epsfxsize=.4\textwidth
\leavevmode
\epsfbox{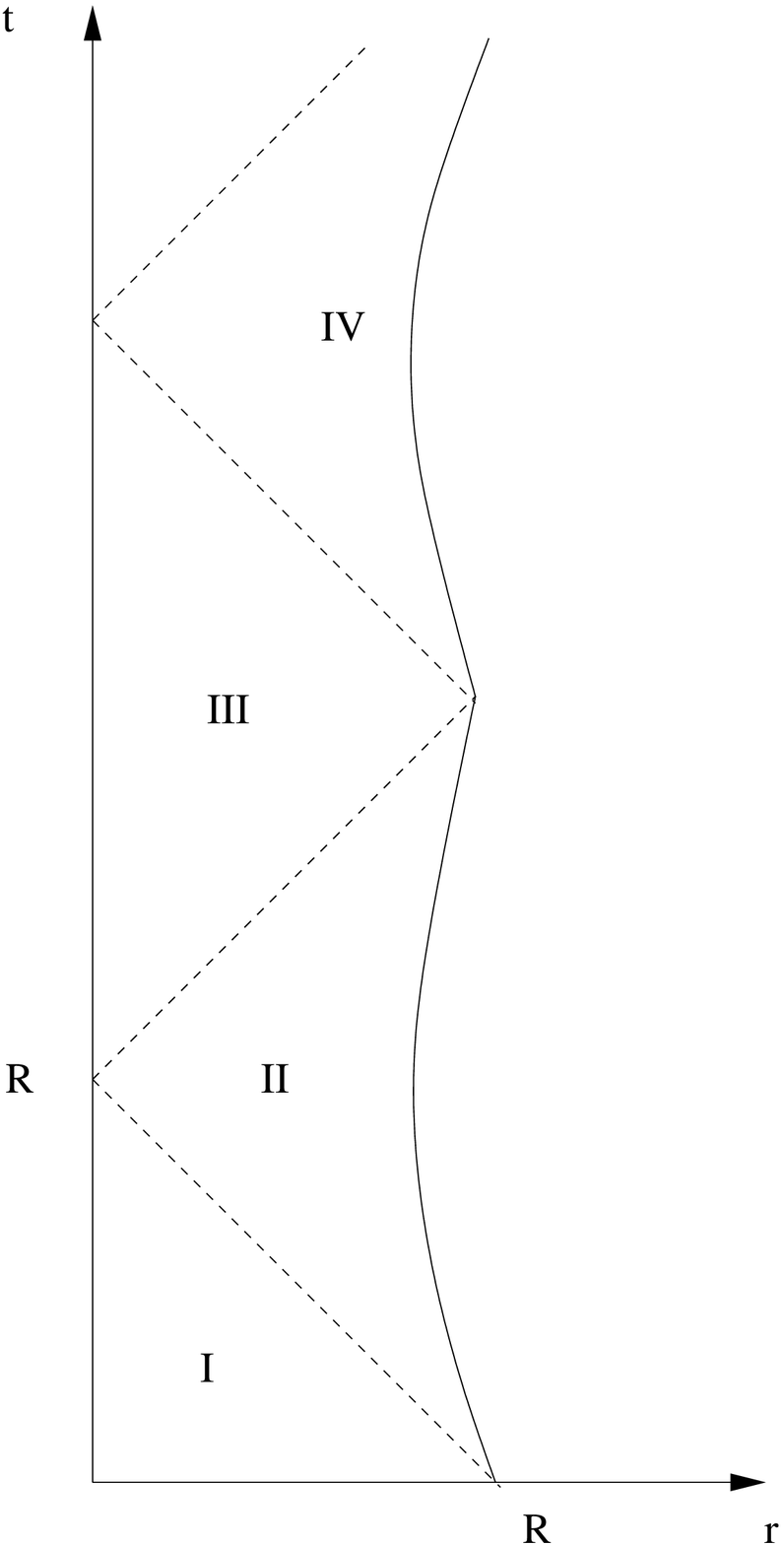}
\end{center}
\caption{Constructing the motion of a rigid elastic sphere.}\label{sphere}
\end{figure}

We linearize $F$ around the origin by computing $F(\varepsilon\varphi, \varepsilon\psi)$ to first order in $\varepsilon$. From
\begin{equation}
\phi(u,v) = - \frac12(u+v) + \frac{\varepsilon\varphi(u)-\varepsilon\psi(v)}{r(u,v)}
\end{equation}
we have
\begin{align}
d\phi = & \left[-\frac12 + \frac{\varepsilon\varphi'(u)}{r} + \frac{\varepsilon\varphi(u)-\varepsilon\psi(v)}{2{r}^2} \right] du \nonumber \\
& + \left[-\frac12 - \frac{\varepsilon\psi'(v)}{r} - \frac{\varepsilon\varphi(u)-\varepsilon\psi(v)}{2{r}^2}\right] dv,
\end{align}
and so
\begin{equation}
F(\varepsilon\varphi, \varepsilon\psi)(v)=\frac14 + \frac{\varepsilon\psi'(v)-\varepsilon\varphi'(v-2R)}{2R}+O(\varepsilon^2).
\end{equation}
In other words,
\begin{equation} \label{84}
(dF(0,0)(\varphi,\psi))(v)= \frac{\psi'(v)-\varphi'(v-2R)}{2R}.
\end{equation}
Note that in particular the kernel of $dF(0,0)(0,\cdot)$ is trivial, since $\psi(R)=0$. On the other hand, \eqref{84} implies that any $\omega \in Z$ is of the form $\omega = dF(0,0)(0,\psi)$ with
\begin{equation}
\psi(v)=2R \int_R^v \omega(s) ds.
\end{equation}
Therefore $dF(0,0)(0,\cdot):Y \to Z$ is an isomorphism. From the Implicit Function Theorem we then conclude that given $\varphi \in U \subset X$ there exists a unique function $\psi \in V \subset Y$ such that $F(\varphi,\psi)=\frac14$ (if the neighborhoods of the origin $U \subset X$ and $V \subset Y$ are taken sufficiently small). 

Recalling Equation~\eqref{phimap}, we then have the following result:

\begin{Thm} \label{freebd}
If $\phi(\cdot,R):[0,R] \to \bbR$ is such that $(r\phi)(\cdot,R)$ is sufficiently close (in the $C^2$ norm) to $(r\phi_{\rm{rel}})(\cdot,R)$ then there exists $\delta > 0$ such that $\phi$ can be continuously extended to a $C^2$ function $\phi:[0,R] \times [R,R+\delta] \to \bbR$ satisfying
\begin{equation}
\Box \phi = 0
\end{equation}
and
\begin{equation}
(\partial_u\phi \, \partial_v\phi)(u_b(v),v) = \frac14,
\end{equation}
where $u_b:[R,R+\delta]\to [-R,R]$ is determined by
\begin{equation}
\begin{cases}
\displaystyle
\frac{du_b}{dv}(v)=\frac{\partial_v\phi}{\partial_u\phi}(u_b(v),v) \\
u_b(R)=-R
\end{cases} 
\end{equation}
and satisfies $u_b(R+\delta) = R$, with $u_b(v) < R$ for $v \in [R,R+\delta]$. 

This extension is unique for $r\phi$ in a sufficiently small $C^2$ neighborhood of $r\phi_{\rm{rel}}$.
\end{Thm}

\begin{proof}
The Implicit Function argument gives $\phi$ as above in $[0,R] \times [R,R+\delta_0]$ for some $\delta_0>0$, but with $u_b(R+\delta_0) < R$. Starting with the values of $\phi(u,R+\delta_0)$ for $u\in[0,R]$, we can use the same argument with the initial null hypersurface $v=R$ replaced with $v=R + \delta_0$. By decreasing $\delta_0$ slightly, it is easy to show that $\phi$ is necessarily $C^2$ across $v=R+\delta_0$. Since we have control over the $C^2$ norm of $r\phi$, we can iterate this procedure until the free boundary reaches $u=R$ (otherwise, if $u_b(R+\delta) < R$, we could extend $\phi$ as before).
\end{proof}

In other words, if we are given $\phi$ on the null line separating regions $I$ and $II$ in Figure~\ref{sphere}, and $r\phi$ is sufficiently close to $r\phi_{\rm rel}$ in the $C^2$ norm, then we can solve the free boundary problem in region $II$. Notice that this yields the values of the function $f$ (see \eqref{problemuv}) in the interval $[R,R+\delta]$, and hence the solution of the wave equation in region $III$. Moreover, since we have the values of $\phi(u,R+\delta)$, we can again apply Theorem~\ref{freebd} to solve the free boundary problem in region $IV$, and so on, for any finite number of these regions (say until we reach $t=T$ for some $T>0$). This construction is unique, in a suitable neighborhood of $\phi_{\rm rel}$, in view of Theorem~\ref{freebd} and the continuity across the shocks imposed by the Rankine-Hugoniot conditions. Therefore we have the following result:


\begin{Thm} \label{5.1}
Let $B \subset \bbR^3$ be the closed ball of radius $R>0$, and let $\phi_0 \in C^2(B)$ and $\phi_1 \in C^1(B)$ be spherically symmetric functions. Given $T>0$, there exist neighborhoods of $\phi_{\rm rel}(0,\cdot)$ in $C^2(B)$ and of $\partial_t\phi_{\rm rel}(0,\cdot)$ in $C^1(B)$ such that for $\phi_0$ and $\phi_1$ in such neighborhoods the initial value free boundary problem written in spherical coordinates as
\begin{equation}
\begin{cases}
\Box\phi = 0 \\
\phi(0,r) = \phi_0(r) \\
\partial_t\phi(0,r) = \phi_1(r) \\
\left((\partial_t\phi)^2-(\partial_r\phi)^2\right)(t,r_b(t)) = 1 \\
\displaystyle
\frac{dr_b}{dt}(t)=-\frac{\partial_r\phi}{\partial_t\phi}(t,r_b(t)) \\
r_b(0) = R
\end{cases}
\end{equation}
has a weak solution $\phi$ defined in the set $S$ given in spherical coordinates by the conditions $0 \leq t \leq T$ and $r\leq r_b(t)$, which is unique in a suitable neighborhood of $\phi_{\rm rel}$. This solution is sectionally $C^1$ in $S$ and sectionally $C^2$ for $r > 0$, with $\grad \phi$ timelike and possibly discontinuous along the null hypersurface generated by the ingoing null rays starting at $t=0, r=R$ upon reflection on the free boundary $r=r_b(t)$, which is sectionally $C^2$.
\end{Thm}

Let us finish by identifying necessary conditions on the initial data for the formation of shocks. Recall that shocks correspond to null hypersurfaces where $\phi$ is not $C^2$. From \eqref{problemuv} we see that $\phi$ can only fail to be $C^1$ at $v=R$ if $f'(v)$ is discontinuous at $v=R$. In this case, $\partial_v \phi(u,v)$ is also discontinuous at $v=R$, and consequently so is $n^2$. Since by construction $n^2=1$ along the free boundary, the initial data must satisfy $n^2(0,R) \neq 1$, that is, the ball cannot be relaxed at the free boundary initially. Similarly, $\phi$ can only fail to be $C^2$ at $v=R$ if $f''(v)$ is discontinuous at $v=R$. In this case, $\partial^2_v \phi(u,v)$ is also discontinuous at $v=R$, and consequently so is any derivative of $n^2$ in any timelike direction. In particular, the derivative of $n^2$ along the free boundary does not vanish initially. Conversely, if $n^2(0,R) \neq 1$ or if the derivative of $n^2$ along the free boundary does not vanish initially then $\phi$ cannot be $C^2$ and we must have a shock. These conditions can easily be translated into conditions on the initial data (see \eqref{n2initial} and \eqref{dn2initial} below). Since these are open conditions, the formation of shocks is stable. When these shocks are present, they focus at the center of symmetry and are reflected on the free boundary, which fails to be $C^2$, and may in fact fail to be $C^1$ (since the velocity field is not $C^1$ and may fail to be continuous).

\begin{Thm}
Under the hypotheses of Theorem~\ref{5.1}, shocks are present if and only if either 
\begin{equation} \label{n2initial}
n^2(0,R) \equiv  (\phi_1(R))^2 - (\phi_0'(R))^2 \neq 1
\end{equation}
or
\begin{align}
& \left[ \left( - \partial_t \phi \, \partial_t + \partial_r \phi \, \partial_r \right) n^2 \right] (0,R) \equiv \nonumber \\
& \qquad - \phi_1(R) \left[ 2 \phi_1(R) \left( \phi_0''(R) + \frac2R \phi_0'(R)\right) - 2 \phi_0'(R) \phi_1'(R) \right]  \label{dn2initial} \\
& \qquad + \phi_0'(R) \left[ 2 \phi_1(R) \phi_1'(R) - 2 \phi_0'(R) \phi_0''(R) \right] \neq 1. \nonumber
\end{align}
In particular, the formation of shocks is stable.
\end{Thm}
%
%
\section*{Acknowledgements}
We thank Moritz Reintjes for discussions in the early stages of this project. This work was partially supported by FCT/Portugal through UID/MAT/04459/2013 and grant (GPSEinstein) PTDC/MAT-ANA/1275/2014.
%
%
\appendix
%
%
\section{Example of a material singularity}\label{appendixA}
In this appendix we show that material singularities may indeed occur in the rigid elastic ball. Consider the smooth spherically symmetric solution of the wave equation
\begin{equation}\label{example}
\phi(t,r)=-\frac{\sin(t)\sin(r)}{r}.
\end{equation}
We have
\begin{equation}
d\phi(t,r)=-\frac{\cos(t)\sin(r)}{r} dt + \frac{\sin(t)(\sin(r) - r \cos(r))}{r^2} dr,
\end{equation}
and so
\begin{equation}
d\phi(0,r)=-\frac{\sin(r)}{r} dt
\end{equation}
is timelike and future-pointing for $r<\pi$. However,
\begin{equation}
d\phi(t,0)= -\cos(t) dt
\end{equation}
becomes null for $t=\frac{\pi}2$. Therefore, if one takes $d\phi(0,r)$ as initial data for an elastic sphere of radius $R \in \left]\frac{\pi}2, \pi\right[$ then a material singularity will form at the center for $t=\frac{\pi}2$. In fact, it forms along the spacelike surface defined by
\begin{equation}
\cotan(t)= \frac1r -\cotan(r), \qquad 0 \leq t \leq \frac{\pi}2, \quad 0 \leq r \leq \pi
\end{equation}
(see Figure~\ref{material_singularity}). Physically, this solution corresponds to starting with a sphere that is progressively more stretched from the center towards the boundary. In the domain of dependence of the initial data, the free boundary does not affect the evolution, and so the outer layers keep stretching the inner layers untill the sphere tears itself apart.

\begin{figure}[h!]
\begin{center}
\includegraphics[scale=.8]{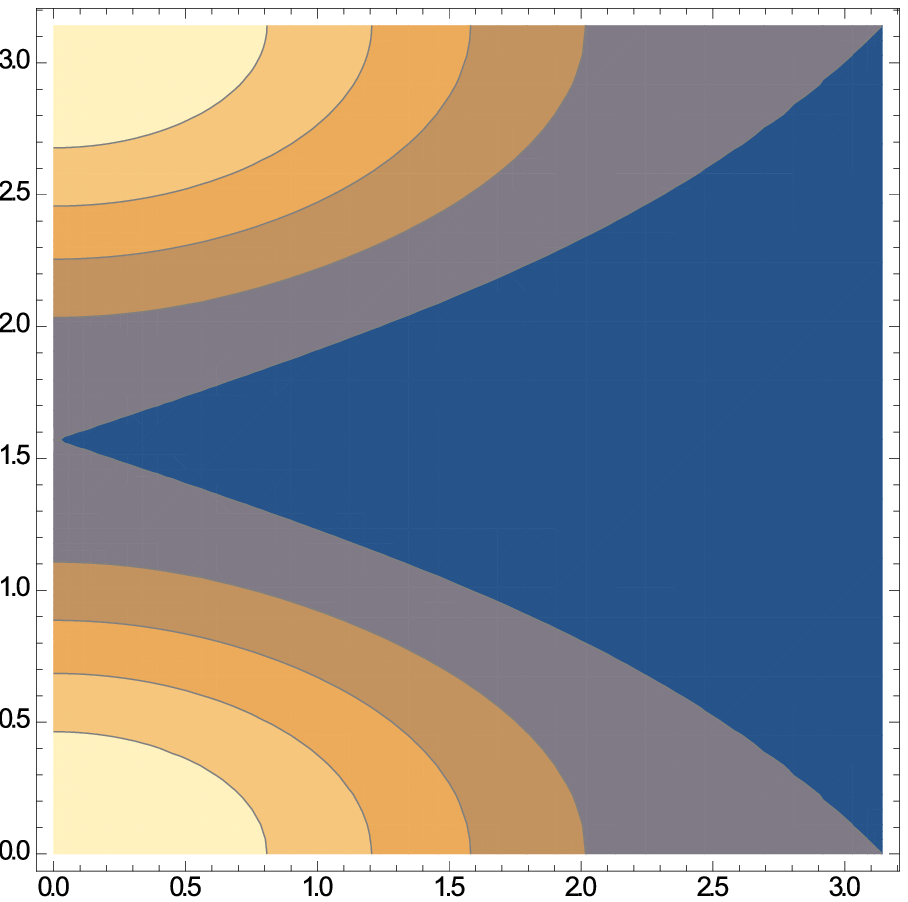}
\includegraphics[scale=.8]{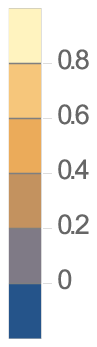}
\end{center}
\caption{Plot of $n^2=-\left\langle d\phi, d\phi \right\rangle$ for $\phi(t,r)$ given in~\eqref{example}. The material singularity forms along the grey-blue transition.}\label{material_singularity}
\end{figure}
%
%
\section{Elastic media}\label{appendixB}
Recall that an elastic medium on a given time-oriented Lorentzian manifold $(M,g)$ is usually modeled by a submersion $\Psi:M \to \bbR^3$ whose level sets are the particle's worldlines (see for instance \cite{BS03, KS03, Wernig06}). The deformation of the medium is then measured by comparing two Riemannian metrics in the hyperplanes orthogonal to the worldlines: the metric $h$ induced by the spacetime metric,
\begin{equation}
h_{\mu\nu} = g_{\mu\nu} + U_\mu U_\nu,
\end{equation}
and the pull-back $\Psi^* \delta$ of the Euclidean metric $\delta$. If the medium is homogeneous and isotropic then the Lagrangian density can only depend on the unordered triple $\{h_1,h_2,h_3\}$ of eigenvalues of $h$ taken with respect to $\Psi^*\delta$.

The Lagrangian density coincides with the energy density measured by the medium's particles, and its choice as a function of $\{h_1,h_2,h_3\}$ is called an {\em elastic law}. Elastic laws that depend only on the number density
\begin{equation}
n^2 = \frac1{h_1 h_2 h_3}
\end{equation}
are usually associated to fluids, since there is no resistance to deformations that preserve the volume. In that sense, one could say that the elastic medium studied in Section~\ref{section5} is just a fluid, although we do allow zero and negative pressures, which are usually ruled out for realistic fluids.

Interestingly, however, spherically symmetric motions of a rigid elastic fluid are actually motions of any elastic body with an elastic law of the form
\begin{equation}
\rho = f(n^2,\sigma^2),
\end{equation}
as long as
\begin{equation}
f(n^2,0) = \frac12 \left( n^2 + 1 \right)
\end{equation}
and the so-called {\em shear scalar} is chosen as
\begin{equation}
\sigma^2 = \left(\frac{h_1}{h_2} - 1\right)^2 \left(\frac{h_2}{h_3} - 1\right)^2 \left(\frac{h_3}{h_1} - 1\right)^2.
\end{equation}
In fact, because we are dealing with a spherically symmetric model, we have to take $\Psi:M \to \bbR^3$ to be a spherically symmetric map, and so $h$ will have two coinciding (tangential) eigenvalues with respect to $\Psi^*\delta$, implying $\sigma^2=0$, that is, $\sigma^2$ does not change under spherically symmetric deformations. Therefore, it is fair to claim that we are indeed considering motions of an elastic medium.
%
%

\end{document}